\documentclass[12pt, reqno]{amsart}
\usepackage{amsmath, amsthm, amscd, amsfonts, amssymb, graphicx, color}
\usepackage[bookmarksnumbered, colorlinks, plainpages]{hyperref}

\usepackage{mathtools}
\usepackage{comment}

\DeclarePairedDelimiter\ket{\lvert}{\rangle}
\DeclarePairedDelimiterX\braket[2]{\langle}{\rangle}{#1\,\delimsize\vert\,\mathopen{}#2}
\newtheorem{theorem}{Theorem}[section]
\newtheorem{remark}{Remark}[section]

\newtheorem{lemma}[theorem]{Lemma}

\theoremstyle{definition}

\newtheorem{conjecture}[theorem]{Conjecture}

\theoremstyle{remark}
\numberwithin{equation}{section}

\begin{document}
\setcounter{page}{1}

\noindent {\small }\hfill     {\small  }\\
{\small }\hfill  {\small }

\centerline{}

\centerline{}

\title[ ]{Aaronson-Ambainis Conjecture Is True For Random Restrictions}

\author[Sreejata Kishor Bhattacharya]{Sreejata Kishor Bhattacharya
 }

\address{School of Technology and Computer Science, Tata Institute of Fundamental Research, Mumbai
}
\email{\textcolor[rgb]{0.00,0.00,0.84}{sreejata.bhattacharya@tifr.res.in}}


\newcommand{\cube}{ \{ \pm 1\} }
\newcommand{ \infl}{ \mathsf{Inf}}
\newcommand{ \prob} {\mathsf{Pr}}
\newcommand{ \var} {\mathsf{Var}}
\newcommand{\E}{\mathsf{E}}
\newcommand {\poly}{\mathsf{poly}}
\begin{abstract}

In an attempt to show that the acceptance probability of a quantum query algorithm making $q$ queries can be well-approximated almost everywhere by a classical decision tree of depth $\leq \poly(q)$, Aaronson and Ambainis proposed the following conjecture: let $f: \cube^n \rightarrow [0,1]$ be a degree $d$ polynomial with variance $\geq \epsilon$. Then, there exists a coordinate of $f$ with influence $\geq \text{poly} (\epsilon, 1/d)$. \newline

We show that for any polynomial $f: \cube^n \rightarrow [0,1]$ of degree $d$ $(d \geq 2)$ and variance $\var[f] \geq  1/d$, if $\rho$ denotes a random restriction with survival probability $\dfrac{\log(d)}{C_1 d}$,
$$ \prob \left[f_{\rho} \text{ has a coordinate with influence} \geq \dfrac{\var[f]^2 }{d^{C_2}} \right] \geq \dfrac{\var[f] \log(d)}{50C_1 d}$$
where $C_1, C_2>0$ are universal constants. Thus, Aaronson-Ambainis conjecture is true for a non-negligible fraction of random restrictions of the given polynomial assuming its variance is not too low. \newline 
\end{abstract} \maketitle

\section{Introduction}

\noindent One of the central open problems in the field of quantum query complexity is finding if there exists a partial function which is defined on a large fraction of the Boolean hypercube (say, constant) but whose quantum query complexity and classical query complexity are super-polynomially separated. The seminal result of Beals, Burhman et al. \cite{BBCMdW98} shows that no such separation is possible when the function is defined on the entire hypercube. On the other hand, functions for which we know such a separation (e.g. -  Forrelation \cite{AA14a} , Bernstein-Vazirani \cite{BV97}) are defined on an exponentially small fraction of the hypercube. A possible explanation as to why all known functions exhibiting large gaps between quantum and classical query complexity have very small support size would be the following folklore conjecture:

\begin{conjecture}
\label{conj:simulatequantum}
    Let $Q$ be a quantum query algorithm with Boolean output on $n$ qubits making $q$ queries. Let $P: \cube^n \rightarrow [0,1]$ be given by $P(x) = \text{Pr}[Q \text{ outputs 1 on }x]$. For any $\epsilon > 0$, there exists a classical query algorithm $A$ such that $\E [(A(x)-Q(x))^2] \leq \epsilon $ and $A$ makes at most $\poly\left(q, \dfrac{1}{\epsilon} \right)$ queries.
\end{conjecture}

It is known that if $Q$ makes at most $q$ queries, then $P$ is given by a polynomial of degree at most $2q$. Although $P$ has more structure than \text{any} arbitrary low degree bounded polynomial, it is further conjectured that such structure is not necessary. In other words, we forget the fact that $P$ arises from a quantum query algorithm and instead try to construct a classical query algorithm for \textit{any} bounded low-degree polynomial. This led to the following conjecture (also folklore).

\begin{conjecture}
\label{conj:approxpoly}
    Let $P: \cube^n \rightarrow [0,1]$ be a degree $d$ polynomial. For any $\epsilon > 0$, there exists a classical decision tree $T$ of depth at most $\poly(d, 1/\epsilon)$ such that $\E [(P(x) - T(x))^2] \leq \epsilon$.
\end{conjecture}

Aaronson and Ambainis \cite{AA14b} proposed the following query algorithm to estimate $P$: suppose the variance of the function is sufficiently small. Then we terminate the query algorithm and output the average over the unqueried coordinates. If not, we query the coordinate with the highest \textit{influence} and restrict the function according to the response received. We keep doing this until we have made too many queries or the variance has become sufficiently low. In order to show that this algorithm gives an accurate estimate, \cite{AA14b} observed that it 
is sufficient to prove the following conjecture.

\begin{conjecture}
\label{conj: AA-conj}
    \textbf{(Aaronson-Ambainis conjecture)} Let $f: \cube^n \rightarrow [0,1]$ be a degree $d$ polynomial. Then, there exists a coordinate $j$ such that $\infl_j[f] \geq \poly (1/d, \var[f])$
\end{conjecture}

As a side remark, we mention that O'Donnell et al. \cite{OSSS05} had shown previously that functions which can be approximated by decision trees have a coordinate with high influence. So conjectures \ref{conj:approxpoly} and \ref{conj: AA-conj} are equivalent. \newline

Aaronson-Ambainis conjecture has received significant attention in the past few years. A 2006 result of Dinur, Friedgut, Kindler, O'Donnell \cite{DFKO06} shows that the conjecture is true if $\poly(d)$ is replaced by $\exp(d)$. In 2012, Montanaro \cite{Mon12} proved the conjecture in the special case of block-multilinear forms where all coefficients have the same magnitude. In 2016, O'Donnell and Zhao \cite{OZ16} showed that it suffices to prove the conjecture for a special class of polynomials known as \textit{one-block decoupled polynomials}. In 2020, Keller and Klein \cite{KK19} claimed to have found a proof for the conjecture but their paper had a subtle flaw and turned out to be wrong. More recently, Lovett and Zhang \cite{LZ19} initiated a new line of attack using the notions of \textit{fractional block sensitivity} and \textit{fractional certificate complexity}. In 2022, Bansal, Sinha, Wolf \cite{BSdW22} proved that this conjecture is true for \textit{completely bounded block multilinear forms} - a class of polynomials that captures a special kind of quantum query algorithms. \newline 

In this work we show that Aaronson-Ambainis conjecture is true for a large fraction of random restrictions of $f$ assuming $\var[f]$ is not too low. We hope our result gives new insights to the Aaronson-Ambainis conjecture. In particular, this opens up a possible line of attack:
\begin{itemize}
    \item Assuming a supposed counterexample $f: \cube^n \rightarrow [0,1]$, modify it appropriately (e.g., by composing it with some appropriate gadget or applying a low noise operator) to get a function $\tilde{f}: \cube^{\tilde{n}} \rightarrow [0,1]$ such that most of its random restrictions remain a counterexample. Combined with our result, this will prove Aaronson-Ambainis conjecture. This approach is discussed in a bit more detail in the conclusion.
\end{itemize}

Our main result is a new structural restriction about bounded low-degree polynomials over the hypercube. While several structural results are known about low-degree \textit{boolean} functions $f: \cube^n \rightarrow \{0,1\}$, such results are rare for low-degree \textit{bounded} functions $f: \cube^n \rightarrow [0,1]$. We show that if $f: \cube^n \rightarrow [0,1]$ has degree $d$ and $\rho$ is a random restriction with survival probability $O(\log(d)/d),$ then with very high probability $f_{\rho}$ depends essentially on $\approx \poly(d)$ coordinates, even though there are $O(n\log(d)/d)$ alive coordinates on average.

\section{Organization}
We introduce notations and necessary preliminaries in section 3. We give a high level overview of our proof in section 4. In section 5 we compile some lemmas that will be needed in our main proof. Our main results are proven in section 6. Our main technical tool is Theorem \ref{theorem:approx_junta}, which says that most random restrictions of a bounded low-degree function can be approximated by a small junta. In Theorem \ref{theorem:aa_for_restrictions} we prove the result mentioned in the abstract (that Aaronson-Ambainis conjecture is true for a non-negligible fraction of random restrictions).
 
\section{Notations and preliminaries}

\subsection*{Query algorithms}
\begin{enumerate}
    \item A classical query algorithm $A$ (or equivalently, a decision tree) for computing a function $f: \cube^n \rightarrow \mathbb{R}$ can access the input $x \in \cube^n$ by adaptively issuing queries to its bits. We assume internal computations have no cost. The depth of the query algorithm/decision tree is the maximum number of bit queries issued on an input. We say $A$ $\epsilon$-approximates $f$ if $||f-A||_2^2 = \E_{x \in \cube^n} \left[ (f(x)-A(x))^2\right] \leq \epsilon$. \newline

    For a partial function $f: S (\subseteq \cube^n) \rightarrow \{0,1\}$, its classical query complexity $D(f)$ is the smallest $d$ for which there exists a decision tree $T$ of depth $d$ such that $T(x)=f(x)$ for all $x \in S$.
    \item A quantum query algorithm can access the input $x \in \cube^n$ via an oracle $O_x$. The oracle acts on a fixed set of $\lceil \log(n) \rceil $ qubits (which the query algorithm has access to) in the following manner:
    $$ O_x  \ket{j}  = (-1)^{x_j}  \ket{j} \text{ for all  }j \in [n].$$
    The quantum query algorithm applies a sequence of unitary operators, where each operator is either $O_x$ or an input-independent unitary operator $U$. In the end, it measures the first qubit and outputs the measurement result. The number of queries issued is the number of times $O_x$ is applied. \newline

    Notice that a quantum query algorithm $Q$ naturally defines a function $ P: \cube^n \rightarrow [0,1]$:
    $$ P(x)= \Pr [Q \text{ outputs }1 \text{ on input }x]$$
    It is well-known that if $Q$ makes $q$ queries, then $P$ is a degree $2q$ polynomial. \newline

    For a function $f: S (\subseteq \cube^n) \rightarrow \{0,1\}$, we define its quantum query complexity $Q(f)$ to be the smallest $q$ for which there exists a quantum query algorithm $Q$ making $q$ queries such that for all $x \in S$,
    \begin{align*}
        \Pr [Q \text{ outputs }1 \text{ on input }x]  \begin{cases}
            \geq 2/3 & \text{ if }f(x)=1\\
             \leq 1/3 & \text{ if }f(x)=0
        \end{cases}
    \end{align*}
    
\end{enumerate}

\subsection*{Analysis of boolean functions}
In this section we recall some results from analysis of boolean functions. A good reference is O' Donnell's textbook \cite{OD21}. \newline

\begin{enumerate}
    \item Any function $f: \cube^n \rightarrow \mathbb{R}$ has a unique representation as $f(x) = \displaystyle \sum_{S \subseteq [n]} \hat{f}(S) \chi_S(x)$ where $\chi_S(x) = \displaystyle \prod_{i \in S} x_i$. The coefficients $\hat{f}(S)$ are the Fourier coefficients of $f$. The degree of $f$ is $\max \{|S| | \hat{f}(S) \neq 0 \}$. \newline
    \item The variance of $f$ is
    $$ \var_{x \in \cube^{n}} [f(x)] = \displaystyle \sum_{S \neq \phi} \hat{f}(S)^2 $$
    \item For a cooordinate $i$, the influence of the $i$'th coordinate is defined as
    \begin{align*}
        \infl_i[f] & = \E_{x \in \cube^n} \left[  \left(\dfrac{f(x) - f(x^{(i)})}{2} \right)^2\right] \\
        & = \displaystyle \sum_{i \in S} \hat{f}(S)^2
    \end{align*}
    The total influence of $f$ is 
    $$ \infl[f] = \displaystyle \sum_{i \in [n]} \infl_i[f] = \displaystyle \sum_{S} |S| \hat{f}(S)^2 $$
    From the Fourier expansion it is clear that if $\text{deg}(f) \leq d, \infl[f] \leq d \var[f]$
    \item Given two functions $f,g \cube^n \rightarrow \mathbb{R}$, we say $g$ $\epsilon-$approximates $f$ if $||f-g||_2^2 =  \E_{x} \left[ (f(x)-g(x))^2 \right] \leq \epsilon$.

    \item For a point $x \in \cube^n$ and a subset $S \subseteq [n]$ and $-1 \leq \rho \leq 1$, we define a distribution $N_{\rho,S}(x)$ on $\cube^n$ as follows:
    \begin{itemize}
        \item The bits $y_1, y_2, \cdots , y_n$ are independent, and
        \begin{align*}
            \prob [y_i= x_i] = \begin{cases}
                1 & \text{ if }i \not \in S \\
                (1+\rho)/2 & \text{ if } i \in S
            \end{cases}
        \end{align*}
    \end{itemize}
    When $S=[n]$, we abbreviate $N_{\rho, S}(x)$ by $N_{\rho}(x)$. \newline
   \item  For $f: \cube^n \rightarrow \mathbb{R}$ and $-1 \leq \rho \leq 1$ define $T_{\rho}f: \cube^n \rightarrow \mathbb{R}$ by 
    $$ T_{\rho}f(x) = \E_{z \leftarrow N_{\rho}(x)} [f(z)].$$
    It is easy to see that the Fourier expansion of $T_{\rho}f$ is given by
    $$ T_{\rho}f(x) = \displaystyle \sum_{S \subseteq [n]} \rho^{|S|} \hat{f}(S) \chi_S(x).$$

    \item A function $f: \cube^n \rightarrow \mathbb{R}$ is a \textit{junta of arity} $l$ or \textit{$l$-junta} if there exists a subset $S \subseteq [n], |S| \leq l$ such that $f$ only depends on the coordinates in $S$. We say $f$ is a $(\epsilon,l)$ junta if it can be $\epsilon$-approximated by a $l$-junta, i.e., there exists a $l$-junta $g$ such that $||f-g||_2^2 \leq \epsilon$.

    \item A restriction $\rho= (S, y)$ of $f: \cube^n \rightarrow \mathbb{R}$ is specified by a subset $S \subseteq [n]$ and an assignment $y \in \cube^{[n] \setminus S}$. Such a restriction naturally induces a function $f_{\rho}: \cube^{S} \rightarrow \mathbb{R}$. Sometimes we shall write $f_y$ instead of $f_{\rho}$ (note that $S$ is determined by $y$ since $y \in \cube^{[n] \setminus S}$) \newline 
    
    By a random restriction with survival probability $p$, we mean sampling $\rho= (S, y \in \cube^{[n] \setminus S})$ where each coordinate $i \in [n]$ is included in $S$ with probability $p$ independently, and each bit of $y$ is independently set to a uniformly random bit. 
    
\end{enumerate}

\begin{remark}
\label{remark:constant}
    Throughout the paper, all growing parameters (e.g., the degree $d$) will be assumed to be larger than some sufficiently big constant. This is to make the expressions look neat, as we will be replacing terms like $(C_1)^k \text{poly}(k)$ by $(C_2)^k$ where $C_2>C_1$. 
\end{remark}
\section{Proof Overview}

The main technical tool in this paper is a structural result for bounded low-degree functions similar in spirit to Hastad's switching lemma \cite{hastad}. Let $f: \cube^n \rightarrow [0,1]$ be a polynomial of degree $d$, and let $\rho$ denote a random restriction with survival probability $\dfrac{\log(d)}{Cd}$. We show that for some constant $C$,
$$ \Pr \left[ f_{\rho} \text{ is a } (O(d^{-C}), O(d^{C})) \text{ junta} \right] \geq 1- \dfrac{1}{d^{\Omega(1)}}.$$

Once this is established, we can prove Aaronson-Ambainis conjecture for random restrictions as follows: it is easy to see that 
$$ \prob \left[  \var [f_{\rho}] \geq \dfrac{\var [f] \log(d)}{2 C d}\right] \geq \dfrac{\var [f] \log(d)}{2 C d}.$$
This probability will be significantly more than the failure probability of the switching lemma  $ \left( \dfrac{1}{d^{\Omega(1)}} \right)$ (this is the only place where we need the lower bound on $\var [f]$). So for a $\approx O(  \var[f] \log(d)/d)$ fraction of random restrictions, the variance of $f_{\rho}$ is high \textit{and} $f_{\rho}$ can be well approximated by a junta with arity $\poly(d)$. This means one of the coordinates of the junta must have high influence. This concludes the proof. \newline

Now we give a brief overview of how we prove the switching lemma. The starting point is the work by Dinur, Friedgut et al. \cite{DFKO06} which states the following:

\begin{theorem}
\label{thm:dinur_original}
For any $f: \cube^n \rightarrow [0,1]$, if $$\displaystyle \sum_{|S| > k} \hat{f}(S)^2 \leq \exp (-O(k^2 \log(k)/\epsilon)),$$ then $f$ is a $(\epsilon, 2^{O(k)}/\epsilon^2)$ junta.
\end{theorem}
In other words, if the Fourier tail above a certain level $k$ is bounded, then $f$ can be well approximated by juntas of arity roughly $2^{O(k)}$.  \newline

We start with the observation that random restrictions have bounded Fourier tails: if the function has degree $d$ and we make a random restriction with survival probability $\dfrac{\log(d)}{Cd}$, using Chernoff bound we can show that with very high probability the Fourier weight above level $\log(d)$ will be low; around the order of $\exp(- \Omega(C \log(d)))$. If we can manage to bring the Fourier weight above $\log(d)$ small enough so that Theorem \ref{thm:dinur_original} applies, then we will get that $f_{\rho}$ can be well approximated by a $\poly(d)$ junta. Unfortunately, if we try this, it turns out that we have to set the survival probability so low that on expectation the variance of $f_{\rho}$ goes down significantly as well. In other words, while it is true that $f_{\rho}$ can be well-approximated by juntas, it is for a trivial reason that its variance itself is very low. (And moreover, this is also true for functions $f:\cube^n \rightarrow \mathbb{R}$ that are merely $L^2$ bounded i.e, $\E [f^2] \leq 1$, so we would not be using any additional structure that arises from the fact that $f$ is \textit{pointwise} bounded.)\newline

In order to make this approach work, we need to improve the tail bound $\exp (-O(k^2 \log(k))/\epsilon)$. The problematic term is the quadratic $\exp(-O(k^2))$ in the exponential. If the dominant term were to the order of $\exp(-O(k))$ instead, the calculations would go through. Can we hope to increase the tail bound to $\exp (-O(k))$ while paying a cost by increasing the junta arity? Unfortunately, again, this is not possible: \cite{DFKO06} constructs a function which shows that the tail bound is essentially tight upto the $\log(k)$ factor - their function has $||f^{>k}||_2^2 \approx \exp (-\Theta(k^2))$ but approximating it to even $1/3$ accuracy requires reading $\Omega(n)$ coordinates. Our key observation is that the function constructed by \cite{DFKO06} has full degree whereas we are working with random restrictions of a low degree function, so in addition to the fact that Fourier tail of $f_{\rho}$ above level $\log(d)$ is very small, we also know that $f_{\rho}$ has degree $d$. Can we hope to improve the tail bound in Theorem \ref{thm:dinur_original} if we have the additional restriction that the function is of degree $d$? Indeed, this turns out to be true. We prove the following result in Theorem \ref{thm:improved_tail}:.

\begin{theorem}
\label{thm:improved_tail}
    There exists a constant $C$ such that the following is true: \newline
    
    If $f: \cube^n \rightarrow [0,1]$ has degree $d$ and $ \displaystyle \sum_{|S|>k} \hat{f}(S)^2 \leq \dfrac{\epsilon}{C^{k} d^C} $, $f$ is a $(\epsilon , \epsilon^{-2} d^C C^k)$ junta.
\end{theorem}

Below we briefly discuss how we are able to improve the tail bound under the additional degree assumption. Dinur, Friedgut, Kindler, O'Donnell \cite{DFKO06} prove their tail bound by showing the following result (we are omitting the exact quantitative parameters here for reading convenience). 
\begin{theorem}
\label{thm:distance_from_bounded}
Let $h: \cube^n \rightarrow \mathbb{R}$ be a degree $k$ function with $\E[h^2] \leq 1$ (but not necessarily pointwise bounded) which cannot be approximated by $2^{O(k)}$ juntas to accuracy $\mu$. Then, for any function $g: \cube^n \rightarrow [0,1]$, $E [(h-g)^2] \geq \varepsilon$. 
\end{theorem}

To prove Theorem \ref{thm:dinur_original}, \cite{DFKO06} applies Theorem \ref{thm:distance_from_bounded} on the truncated function $h = f^{\leq k} = \displaystyle \sum_{|S| \leq k} \hat{f}(S) \chi_S$ and takes $g$ to be the original function $f$. This then lower bounds $E[(f-f^{\leq k})^2]$ which is precisely the Fourier tail above weight $k$. Thus, the distance lower bound $\varepsilon$ in Theorem \ref{thm:distance_from_bounded} governs the Fourier tail lower bound in Theorem \ref{thm:dinur_original}. Since we have the additional information that $f$ is of degree $d$, for our purposes it will suffice to bound the distance from bounded degree $d$ functions, not necessarily \textit{all} bounded functions. In Theorem \ref{theorem:main_result} we prove a result of the following form (again, we are omitting the exact parameters for reading convenience)

\begin{theorem}
\label{thm:distance_from_low_degree}
Let $h: \cube^n \rightarrow \mathbb{R}$ be a degree $k$ function with $\E[h^2] \leq 1$ (but not necessarily pointwise bounded) which cannot be approximated by $2^{O(k)}$ juntas to accuracy $\mu$. Then, for any \textit{degree $d$ }function $g: \cube^n \rightarrow [0,1]$, $E [(h-g)^2] \geq \tilde{\varepsilon}$. 
\end{theorem}

The parameter $\tilde{\varepsilon}$ in Theorem \ref{thm:distance_from_low_degree} is bigger than the corresponding $\varepsilon$ parameter in Theorem \ref{thm:distance_from_bounded} because we are only lower bounding the distance of $h$ from bounded \textit{low-degree} functions whereas Theorem \ref{thm:distance_from_bounded} lower bounds the distance of $h$ from \textit{arbitrary} bounded functions. It turns out that the improvement in this parameter is sufficiently good for the random restriction approach to go through. \newline

In order to prove Theorem \ref{thm:distance_from_low_degree} we shall use the main idea of the proof of \cite{DFKO06} along with a structural restriction for bounded low-degree functions discovered first in \cite{BBCMdW98}. Given any function $f: \cube^n \rightarrow \mathbb{R}$ and $x \in \cube^n$ define the \textit{block sensitivity} of $f$ at $x$ to be
$$ \text{bs}(f,x) = \sup \left[ \displaystyle \sum_{j \in [k]} \left| f(x) - f(x^{(B_j)}) \right| \right]$$
where the supremum ranges over all partitions $(B_1, B_2, \cdots , B_k)$ of the variables ($x^{(B_j)}$ denotes $x$ with the coordinates in $B_j$ flipped). Define the block sensitivity of $f$ to be $\text{bs}(f) = \sup_x \text{bs}(f,x)$. We shall use the following fact about bounded low-degree functions:

\begin{theorem}
\label{thm:block_sensitivity}
    \cite{BBCMdW98} If $f: \cube^n \rightarrow [0,1]$ has degree $d$, $\text{bs}(f) \leq 6d^2$.
\end{theorem}

We give a high-level overview of how we are able to improve upon the bound of $\varepsilon$ using the fact that the block sensitivity of a bounded low degree function is small. At one point in their proof, \cite{DFKO06} lower bounds the probability of a linear form of Rademacher random variables $l(x_1, x_2, \cdots , x_t) = a_1 x_1 + \cdots + a_t x_t$ exceeding a certain threshold times its standard deviation, i.e.,
$$ \prob \left[ a_1 x_1 + \cdots + a_t x_t \geq \alpha \sqrt{a_1^2 + \cdots + a_t^2}\right].$$
For each such point $x$ where this linear form is high, \cite{DFKO06} shows that \textit{many related points} $x'$ must have $f(x') > 2$. Using this they conclude that $f$ must deviate from the interval $[0,1]$ too often and therefore cannot be approximated by any bounded function. \newline

    We follow the proof of \cite{DFKO06} up until this point. Instead of directly lower bounding the probability that $a_1x_1 + \cdots + a_tx_t$ is high, we partition the set of variables $[t]$ into $L$ blocks $B_1 , \cdots , B_L$ ($L$ is an appropriately chosen parameter) such that each block gets roughly same total weight: for all $j \in [L],$
    $$ \displaystyle \sum_{i \in B_j} a_i^2 \geq \dfrac{a_1^2 + \cdots + a_t^2}{2L}.$$
    It will turn out that the $a_i$'s are sufficiently small for such a partition to exist. For each block we lower bound the probability that the linear form restricted to this block is high:    
$$ \prob \left[\displaystyle \sum_{j \in B_i} a_j x_j \geq \tilde{\alpha} \sqrt{\displaystyle \sum_{j \in B_i} a_j^2} \right]. $$
Now, on a random assignment $z$, the linear form restricted to many of these blocks will be high. Take such a block $B_i$: $\displaystyle \sum_{j \in B_i} a_j z_j \geq \Tilde{\alpha} \sqrt{\displaystyle \sum_{j \in B_i} a_j^2}$. For each such block we will be able to find a large number of related points $z_i$ such that $|f(z) - f(z_i)|$ is large. Crucially, these related points will differ from $x$ only at $B_i$. Thus, we will find many points which differ from $z$ at disjoint sets and whose $f$ differ from $z$ significantly. This will show that $f$ cannot be too close to a bounded low degree function, because those functions have low block sensitivity. \newline

Our advantage is that we need to set $\alpha'$ so that we can conclude $|f(z) - f(z')|$ is only somewhat larger than $\Omega(d^2/L)$ (as opposed to $\Omega(1)$ in \cite{DFKO06}) - by setting $L$ large enough this allows us to set a much smaller $\alpha$ and get rid of the quadratic exponential dependence.
\section{Tools}

In this section we compile some lemmas that we shall use in our proof.
\subsection*{A reverse Markov inequality}
We will use the following simple inequality throughout the proof.

\begin{lemma}
    \label{lemma:reverse-markov}
    Let $X$ be a random variable such that $X \leq M$ with probability 1. Let $\E [X]=  \mu > 0$. Then, $\prob [X \geq \mu/2] \geq \dfrac{\mu}{2M}$
\end{lemma}

\begin{proof}
    Assume $\prob [X \geq \mu/2] < \dfrac{\mu}{2M}$. Then,
    $$ \E [X] \leq \prob [X \geq \mu/2] M + \prob [X \leq \mu/2] \dfrac{\mu}{2} < \mu, $$
    contradiction.
\end{proof}
\subsection*{An anticoncentration inequality for linear forms of Rademacher random variables}
\begin{lemma}
\label{lemma:anticoncentration}
There exists a universal constant $K$ such that the following holds: let $x_1, \cdots , x_n$ be independent Rademacher random variables and let $l(x_1, \cdots , x_n) = a_1 x_1 + \cdots + a_n x_n$. Let $\sigma = \sqrt{a_1^2 + \cdots + a_n^2}$. Suppose $|a_i| \leq \dfrac{\sigma}{Kt}$. Then,
$$ \Pr [l(x_1, \cdots , x_n) \geq t \sigma ] \geq \exp (- K t^2).$$
\end{lemma}
\begin{proof}
    Equation 4.2 in \cite{LT91}.
\end{proof}
\subsection*{Random restrictions have small tail}

\begin{lemma}
\label{lemma:bounded_tail}
    Let $f:\cube^n \rightarrow \mathbb{R}$ have degree $d$, $C>1$ be a sufficiently large constant, and let $\rho = (S, y \in \cube^{[n] \setminus S})$ be a random restriction with survival probability $\dfrac{\log(d)}{Cd}$. Let $k = \log(d)$. Then,
    $$ \E \left[ \displaystyle \sum_{|T| > k} \hat{f_y}(T)^2 \right] \leq \exp(-C \log(d)/8) \var[f].$$
\end{lemma}
\begin{proof}

First suppose $(S, y \in \cube^{[n] \setminus S})$ is a fixed restriction. Note that for $z \in \cube^S,$

$$f_y(z) = \displaystyle \sum_{U \subseteq [n]} \hat{f}(U) \chi_U(y,z)$$
so for $T \subseteq [S], $ $\hat{f_y}(T) = \displaystyle \sum_{U \subseteq S} \hat{f}(U \cup T) \chi_U(y)$. By Parseval's theorem, for a fixed $S$,
$$ \E_y \left[ \displaystyle \hat{f_y} (T)^2 \right] = \displaystyle \sum_{U \subseteq [n] \setminus S} \hat{f}(U \cup T)^2 .$$
Therefore, for a fixed $S$,

$$ \E_{y} \left[ \displaystyle \sum_{|T| > k } \hat{f_y}(T)^2 \right] = \displaystyle \sum_{V \subseteq [n], |V \cap S| > k }  \hat{f}(V)^2.$$
Randomizing over $S$ again, 
$$ \E_{S,y} \left[\sum_{|T| > k } \hat{f_y}(T)^2 \right]= \displaystyle \sum_{V \subseteq [n]} \Pr [|V \cap S| > k] \hat{f}(V)^2.$$

Since $f$ has degree $d$, we only need to worry about the terms where $|V| \leq d$. Also, for $|V| \leq k$ the relevant probability is 0. Since each element is included in $S$ with probability $\dfrac{ \log(d)}{Cd}$, by Chernoff bound, for each $V$ with $|V| \leq d$, 
$$ \Pr \left[ |V \cap S| > k \right] \leq \exp ((C-1)^2k/4C) \leq \exp(-C \log(d)/8).$$
Thus we get that
$$ \E _{S,y} \left[ \displaystyle \sum_{|T| > k} \hat{f_y}(T)^2 \right] \leq \exp(-C \log(d)/8) \displaystyle \sum_{T \neq \phi} \hat{f}(T)^2 = \exp(-C \log(d)/8) \var[f].$$

\end{proof}
\subsection*{Random restrictions don't have low variance}
\begin{lemma}
\label{lemma:variance_of_random_restriction}
    Let $f:\cube^n \rightarrow \mathbb{R}$ be any function and let $\rho = (S, y \in \cube^{[n] \setminus S})$ be a random restriction with survival probability $p$. Then, $\E [\var[f_{\rho}]] \geq p \var[f].$
\end{lemma}

\begin{proof}
    Fix a restriction $(S, y \in \cube^{[n] \setminus S})$. For each $T \subseteq S$, $\hat{f_y}(T) = \displaystyle \sum_{U \subseteq [n] \setminus S} \hat{f}(T \cup U) \chi_U (y)$. Thus, by Parseval's theorem, for a fixed $S$, 
    $$ \E_y [\var[f_y]] = \displaystyle \sum_{T: T \cap S \neq \phi} \hat{f}(T)^2.$$
    Randomizing over $S$ again,
\begin{align*}
    \E_{S,y} [ \var[f_y]] & = \displaystyle \sum_{T \neq \phi} \Pr [T \cap S \neq \phi] \hat{f}(T)^2 \\
    & = \displaystyle \sum_{T \neq \phi} \left( 1 - (1-p)^{|T|} \right) \hat{f}(T)^2 \\
    & \geq \displaystyle \sum_{T \neq \phi} p \hat{f}(T)^2 \\
    & = p \var[f].
\end{align*}
\end{proof}

\subsection*{{Random restrictions with appropriate survival probability put large Fourier mass on the linear level}}

\begin{lemma}
\label{lemma:level_one}
    Let $f: \cube^n \rightarrow \mathbb{R}$, $J \subseteq [n]$ and $k$ be such that 
    $$ \displaystyle \sum_{2^k \leq |T \cap J^c| < 2^{k+1} } \hat{f}(T)^2 \geq \mu .$$
    Consider a random restriction $\rho = (S, y \in \cube^{[n] \setminus S}) $ where each $j \in J$ is fixed and given an uniformly random assignment, and each $i \in J^c $ is kept alive with probability $p=2^{-k}$. Then,
    $$ \E \left[ \displaystyle \sum_{i \in S} \hat{f}_y(\{i\})^2 \right] \geq \dfrac{\mu}{20}.$$
\end{lemma}

\begin{proof}
    For a fixed $(S, y \in \cube^{[n] \setminus S}$ (note that $J \cap S = \phi$) and $j \in S$ we have
    $$ \hat{f_y}(\{j\}) = \displaystyle \sum_{T \subseteq [n], T \cap S = \{j\}} \hat{f}(T) \chi_{T \setminus \{j\}} (y).$$
    By Parseval's theorem, for a fixed $S$,
    $$ \E_y \left[ \hat{f_y}(\{j\})^2  \right] = \displaystyle \sum_{T \subseteq [n], T \cap S = \{j\}} \hat{f}(T)^2 .$$
    Randomizing over $S$,
    \begin{align*}
        \E \left[ \displaystyle \sum_{j \in S} \hat{f_y}(\{j\})^2 \right] & = \displaystyle \sum_{T \subseteq [n]} \left[ \displaystyle \sum_{j \in [n]} \Pr [T \cap S= \{j\}] \right] \hat{f}(T)^2 \\
        & \geq \displaystyle \sum_{2^k \leq |T \cap J^c| < 2^{k+1}} |T|p (1-p)^{|T|-1} \hat{f}(T)^2.
    \end{align*}
    By standard inequalities, for $n \in [1/p, 2/p),$ $np (1-p)^{n-1} \geq 1/20$. It follows that
    \begin{align*}
\E_y \left[ \hat{f_y}(\{j\})^2  \right] & \geq \dfrac{\mu}{20}   .     
    \end{align*}

\end{proof}

\subsection*{Some hypercontractive inequalities for low degree functions}

The proof of these lemmas can be found in \cite{DFKO06}.

\begin{lemma}
\label{lemma:BGB_biased}
    There exists a universal constant $W>0$ such that the following holds: \newline
    
    Let $f: \cube^{n} \rightarrow \mathbb{R}$ be a degree $k$ function. Let $\E [f^2]= \sigma^2$. Then,
    $$ \E \left[ f(z)^2 \mathsf{1}_{f(z)^2 \leq W^k \sigma^2 } \right] \geq \dfrac{1}{2} \E [f(z)^2].$$
\end{lemma}
\begin{proof}
    Corollary 2.4 in \cite{DFKO06}.
\end{proof}
\begin{lemma}
\label{lemma:hypercontractivity}   
    There exists a universal constant $B>0$ such that the following holds: \newline
    
    Let $f: \cube^n \rightarrow \mathbb{R}$ be a degree $k$ function. Let $\rho \in [-1/2, 1/2]$ be a noise parameter, and $x_0 \in \cube^n$. Suppose $\E_{z \leftarrow N_{\rho}(x)} [f(z) -  f(x_0)] = \mu \geq 0$. Then,
    $$ \Pr_{z \leftarrow N_{\rho}} \left[ f(z) - f(x_0)  \geq \mu \right] \geq \dfrac{1}{B^k}.$$
\end{lemma}

\begin{proof}
    Lemma 2.5 in \cite{DFKO06}
\end{proof}

\subsection*{The noise lemma} This is the main result of \cite{DFKO06}. We use a slight variant. First we recall some known results from approximation theory.

\begin{lemma}
\label{lemma:interpolation}
    For any $k$, there exist constants $\rho_1, \rho_2, \cdots , \rho_{k+1} \in [-1/2, 1/2]$ with the following property: for any polynomial of degree $k$, $p(x) = a_0 + a_1 x + \cdots + a_k x^k$, there exists a $j \in [k+1]$ such that $|p(\rho_j)| \geq \dfrac{|a_1|}{2(k+1)}.$
\end{lemma}
\begin{proof}
    Page 112 in \cite{Rivlin}.
\end{proof}

Now we state the lemma.

\begin{lemma}
\label{lemma:noise_lemma}
    There exists a universal constant $B>0$ such that the following holds: \newline
    
    Consider a degree $k$ polynomial $f: \cube^n \rightarrow \mathbb{R}$. Let $S \subseteq [n]$ and $\ell(x) = \displaystyle \sum_{i \in S} \hat{f}(\{i\}) x_i$. Consider an input $x_0 \in \cube^n$ such that $\ell(x_0) \geq \gamma$. Sample a $z \leftarrow \cube^n$ by the following procedure:
    \begin{enumerate}
        \item Sample $\rho \leftarrow \{\rho_1, \cdots , \rho_{k+1} \}$ uniformly at random.
        \item Sample $z \leftarrow N_{\rho, S} (x_0)$.
    \end{enumerate}
    Then,
    $$ \Pr \left[ \left| f(z) - f(x_0) \right| \geq \dfrac{\gamma}{2(k+1)} \right] \geq \dfrac{1}{(k+1)B^k}.$$
    
\end{lemma}

\begin{remark}
    Observe that $z$ differs from $x$ only in the coordinates of $S$. This will be crucial later on.
    
\end{remark}

\begin{proof}
    Take $B$ to be the same universal constant as in Lemma \ref{lemma:BGB_biased}.
    By replacing $f$ with an appropriate restriction if necessary, we can assume $S = [n]$. Consider the polynomial $p(\rho) = T_{\rho}f(x_0) - f(x_0)$. From the Fourier expansion of noise operator, we see that
    $$ p(\rho) = \displaystyle \sum_{S \neq \phi} \rho^{|S|} \hat{f}(S).$$
    This is a degree $k$ polynomial in $\rho$ with linear coefficient $l(x_0)$. By Lemma \ref{lemma:interpolation}, there exists a $h \in [k+1]$ such that $p(\rho_h) \geq \gamma/(2k+2)$. By Lemma \ref{lemma:BGB_biased}, 
    $$ \Pr_{z \leftarrow N_{\rho}(x_0)} \left[\left|f(z) - f(x_0)\right| \geq \dfrac{\gamma}{2(k+1)} \big{|} \rho= \rho_h \right] \geq \dfrac{1}{B^k}.$$
    We choose $\rho= \rho_h$ in step (1) with probability $1/(k+1)$, so
    $$ \Pr_{z \leftarrow N_{\rho}(x_0)} \left[\left|f(z) - f(x_0)\right| \geq \dfrac{\gamma}{2(k+1)} \right] \geq \dfrac{1}{(k+1)B^k}.$$    
\end{proof}

\subsection*{Partitioning a set of numbers in a balanced manner}
We need an easy lemma about partitioning a set of weights none of which is too large into disjoint buckets where each bucket gets roughly the same total weight. We will later use this lemma on the set of small linear Fourier coefficients of a function.
\begin{lemma}
\label{lemma:balanced_partition}
    Let $a_1, a_2, \cdots , a_n$ be a set of non-negative real numbers and $1 \leq L \leq n$. Suppose $a_i \leq \dfrac{a_1 + a_2 + \cdots + a_n}{2L}$ for all $1 \leq i \leq n$. Then, there exists a partition $(B_1, B_2, \cdots , B_L)$ of $[n]$ such that for all $1 \leq j \leq L$,

    $$ \displaystyle \sum_{i \in B_j} a_i \geq \dfrac{a_1 + \cdots + a_n}{2L}.$$
\end{lemma}
\begin{proof}
Start with an arbitrary partition $(B_1, B_2, \cdots , B_L)$. Then, refine it iteratively according to the following algorithm. \newline
    
\fbox{\parbox{\textwidth}{\textit{Refinement algorithm:}
\begin{enumerate}

\item Locate a $j$ such that the condition is violated for $j$, i.e., 
$$ \displaystyle \sum_{i \in B_j} a_i < \dfrac{a_1 + \cdots + a_n}{2L}.$$
If no such $j$ exists, terminate.
\item Locate a $k$ such that
$$ \displaystyle \sum_{i \in B_k} a_i \geq \dfrac{a_1 + \cdots + a_n}{L}.$$
\item Take an arbitrary $l \in B_k$ such that $a_l \neq 0$ and place it in $B_j$;
$$B_k \leftarrow B_k \setminus \{l\}$$
$$B_j \leftarrow B_j \cup \{l\} $$
\end{enumerate}}}
An appropriate $k$ always exists in step (2) by an averaging argument. Since $a_l \leq \dfrac{a_1 + \cdots + a_n}{2L}$, the size of $B_k$ does not go below $\dfrac{a_1 + \cdots + a_n}{2L}$ after step (3). It is easy to see this procedure must terminate. Formally, notice that the quantity
$$ \displaystyle \sum_{j \in [L]} \min \left( \dfrac{a_1 + \cdots + a_n}{2L} - \displaystyle \sum_{i \in B_j} a_i, 0 \right)$$
reduces by $\min \{ a_i | a_i \neq 0 \}/2L$ at each step, so at some point of time it must be 0 at which point the algorithm terminates and returns a valid partition.

\end{proof}

\section{Main results}

\subsection{Improved tail bound for low degree functions}
This section is the core technical part of our work: we show that if we have a function $f: \cube^n \rightarrow \mathbb{R}$ \textit{(not necessarily bounded)} with $\E[f^2] \leq 1$ which cannot be approximated by juntas, then $f$ cannot be well-approximated by bounded low-degree functions. \newline

For a subset $J \subseteq [n]$, consider the junta $u: \cube^n \rightarrow \mathbb{R}$ which reads the coordinates of $J$ and outputs the average over the unqueried coordinates. It is easy to see that 
$u(x) = \displaystyle \sum_{S \subseteq J} \hat{f}(S) \chi_S(x),$ so $ ||u-f||_2^2 = \displaystyle \sum_{S \not \subseteq J} \hat{f}(S)^2.$ Thus, $u$ approximates $f$ if and only if $\displaystyle \sum_{S \not \subseteq J} \hat{f}(S)^2$ is small. 

\begin{remark}
\label{remark:junta_approximation}
    In fact, it is easy to see that there exists a junta $u$ depending only on coordinates of $J$ such that $||f-u||_2^2 \leq \epsilon$ if and only if $\displaystyle \sum_{S \not \subseteq J} \hat{f}(S)^2 \leq \epsilon$. This immediately follows from the Fourier expansion of $f-u$.
\end{remark}

\begin{theorem}
\label{theorem:main_result}

There exists a constant $C$ such that the following holds: \newline
Let $f: \cube^n \rightarrow \mathbb{R}$ be a degree $k$ function \textit{(not necessarily bounded)} with $\E[f^2] \leq 1$. Let $J = \{j | \infl_j[f] \geq \theta \}$ where $\theta = \dfrac{\mu^2}{C^k d^C}$. If $\displaystyle \sum_{S \not \subseteq J} \hat{f}(S)^2 \geq \mu$, then for any degree $d$ function $g: \cube^n \rightarrow [0,1]$, 
$\E [(f(x)-g(x))^2] \geq \delta = \dfrac{\mu}{C^k d^C}$.
\end{theorem}
\begin{remark}
    Notice here that although $f$ is not pointwise bounded, $g$ is.
\end{remark}
\begin{proof}
Let $W, B,K$ be the universal constants from Lemma \ref{lemma:BGB_biased}, Lemma \ref{lemma:noise_lemma} and Lemma \ref{lemma:anticoncentration} respectively. We take $C$ to be a constant sufficiently larger than $B,K,W$. 

There exists a $t$ such that
$$ \displaystyle \sum_{ 2^t \leq |S \cap J^c| < 2^{t+1} } \hat{f}(S)^2 \geq \dfrac{\mu}{\log(k)}.$$ Let $\rho = (U, y \in \cube^{[n] \setminus U})$ be a random restriction where each $j \in J$ is killed and given a uniformly random assignment, and survival probability for each $j \not \in J$ is $2^{-t}$. By Lemma \ref{lemma:level_one},
\begin{align*}
    \E \left[ \displaystyle \sum_{j \in U} \hat{f_y}(\{j\})^2 \right] \geq \dfrac{\mu}{20 \log(k)}.
\end{align*}
Fix a $U$ such that
$$ \E_{y \in \cube^{[n] \setminus U}} \left[ \displaystyle \sum_{j \in U} \hat{f_y}(\{j\})^2 \right] \geq \dfrac{\mu}{20 \log(k)}.$$

By Parseval's theorem we have for all $j \in U$,
$$ \E \left[ \hat{f_y}(\{j\})^2 \right] = \displaystyle \sum_{S \cap U = \{j\} } \hat{f}(S)^2 \leq \infl_j [f].$$
For each $y \in \cube^{[n]}$ define $\mathsf{SMALL}_y = \{j | \hat{f_y}(\{j\})^2 \leq W^k \infl_j[f] \}$. Observe that for all $y$,
$$ \displaystyle \sum_{j \in \mathsf{SMALL}_y} \hat{f_y}(\{j\})^2 \leq W^k \infl[f] \leq k \cdot W^k \leq (2W)^k.$$

For each $j \in U$ we have from Lemma \ref{lemma:hypercontractivity}
$$ \E \left[ \hat{f}_y(\{j\})^2 \mathsf{1}_{\hat{f}_y(\{j\})^2 \leq W^k \infl_j[j] } \right] \geq \dfrac{1}{2}\E \left[ \hat{f}_y(\{j\})^2 \right] . $$
Thus,
$$ \E_{y \in \cube^{[n] \setminus U}} \left[ \displaystyle \sum_{j \in \mathsf{SMALL}_y} \hat{f}_y(\{j\})^2 \right] \geq \dfrac{\mu}{40 \log(k)},$$
so applying Lemma \ref{lemma:reverse-markov} $^{2}$
$$ \prob_{y \in \cube^{[n] \setminus U}} \left[  \sum_{j \in \mathsf{SMALL}_y} \hat{f}_y(\{j\})^2 \geq \dfrac{\mu}{80 \log(k)} \right] \geq \dfrac{\mu}{ 80 \log(k) (2W)^k} \geq \dfrac{\mu}{(3W)^k}.$$
\footnotetext[2]{See remark \ref{remark:constant}.}

Call $y \in \cube^{[n] \setminus U}$ for which $  \sum_{j \in \mathsf{SMALL}_y} \hat{f}_y(\{j\})^2 \geq \dfrac{\mu}{40 \log(k)} $ to be \textit{good}. Let $\text{GOOD} = \{y \in \cube^{[n] \setminus U} | y \text{ is good} \}$. Let  $L = \left \lceil \dfrac{(2B)^k d^8}{\var [f]} \right \rceil$. For each good $y$, choose a partition $\mathsf{DIVIDE}(y) = (B_1, B_2, \cdots , B_L)$ of $\mathsf{SMALL}_y$ such that for all $1 \leq i \leq L$, 
$$ \displaystyle \sum_{j \in B_i} \hat{f}_y(\{j\})^2 \geq \dfrac{\mu}{80 L \log(k)}.$$
(If there are multiple such partitions, choose any one of them and call it $\mathsf{DIVIDE}(y)$.) \newline

Our choice of parameters ensures that for all $j \in \mathsf{SMALL}_y,$ $\hat{f}_y(\{j\})^2 \leq \dfrac{\mu}{80 L \log(k)}$, so such a partition exists by Lemma \ref{lemma:balanced_partition}. Let $\rho_1, \rho_2, \cdots , \rho_{k+1}$ be the constants from Lemma \ref{lemma:interpolation}. \newline

Suppose, for the sake of contradiction, there exists a degree $d$ polynomial $g: \cube^n \rightarrow [0,1]$ such that $\E [(f(x)-g(x))^2] \leq \delta$. Throughout the rest of the proof, for a string $s_1 \in \cube^{[n] \setminus U}$ and a string $s_2 \in \cube^U$, the pair $(s_1, s_2)$ denotes the string $s \in \cube^n$ which agrees with $s_1$ on $[n] \setminus U$ and with $s_2$ on $U$. Consider the following randomized procedure which returns a real number.

\fbox{\parbox{\textwidth}{\textit{Procedure 1:}
\begin{enumerate}

\item Sample a $y \in \mathsf{GOOD}$ uniformly at random.
\item Sample $\rho \leftarrow \{ \rho_1, \rho_2, \cdots , \rho_{k+1} \}$ uniformly at random.
\item Sample $z \leftarrow \cube^{U}$ uniformly at random.
\item Let $\mathsf{DIVIDE}(y) = (B_1, B_2, \cdots , B_L)$. Sample $\Tilde{z}^{(i)} \leftarrow N_{B_i, \rho} (z)$ for $1 \leq i \leq L$.
\item Return $\displaystyle \sum_{i=1}^{L} |f(y,z) - f(y,\Tilde{z}^{(i)})|$.
\end{enumerate}}}
We estimate the probability that procedure 1 returns a number $>15d^2$ in two different ways. First, we obtain a lower bound from the definition of $\text{GOOD}$. Then, we obtain an upper bound from the assumption that $\E [ (f(x)-g(x))^2] \leq \delta$ and the fact that we have a lower bound on $\text{Pr}_y[y \in \text{GOOD}]$ (which, recall, follows from the assumption that $\displaystyle \sum_{S \not \subseteq J} \hat{f}(S)^2 \geq \mu$). These two bounds will contradict each other - and that will prove the theorem.

\subsection*{Lower bound: }
Fix a $y \in \mathsf{GOOD}$. Let $\mathsf{DIVIDE}(y)= (B_1, B_2, \cdots , B_L).$\newline

Let $w =  \sqrt{\dfrac{\mu}{80 L \log(k)}}$. For each $i \in [L]$ we have
$$ \sqrt{ \displaystyle \sum_{j \in B_i} \hat{f}_y(\{j\})^2} \geq w.$$
Choose $\alpha$ such that $\alpha w = \dfrac{100 d^4 (2B)^k}{L}$. Our choice of $L$ ensures that $\alpha \leq 1$. Moreover, our choice for influence threshold $\theta$ ensures that $|\hat{f}_y(\{j\})| \leq \dfrac{w}{K \alpha}$ for all $j \in B_i$ where $K$ is the universal constant from the Lemma \ref{lemma:anticoncentration}. 

Therefore, we can apply Lemma \ref{lemma:anticoncentration} to obtain that
$$ \prob_{z \in \cube^{B_i}} \left[ \displaystyle \sum_{j \in B_i} z_j \hat{f}_y(\{j\}) \geq \dfrac{100 d^4 (2B)^k}{L} \right] \geq \exp (- K \alpha^2) \geq \dfrac{1}{K_1}.$$
Here $K_1 = \exp (K)$ is an absolute constant.

By Lemma \ref{lemma:noise_lemma} applied on the restriction $f_y: \cube^U \rightarrow [0,1]$, as we sample $z \leftarrow \cube^{U}$ u.a.r, $\rho \leftarrow \{ \rho_1, \cdots , \rho_{k+1} \}$ u.a.r, $\tilde{z}^{(i)} \leftarrow N_{\rho, B_i} (z)$, we have that

$$ \prob \left[ |f(y,z) - f(y, \tilde{z}^{(i)}) | \geq \dfrac{30 d^3 (2B)^k}{L} \right] \geq \dfrac{1}{K_1 (k+1) B^k} \geq \dfrac{1}{(2B)^k}.  $$
By linearity of expectation,
$$ \E \left[ \left| \left \{ i \in [L] | | f(y,z) - f(y,\Tilde{z}^{(i)})| \geq \dfrac{30 d^3 (2B)^k}{L}  \right \}\right| \right] \geq \dfrac{L}{(2B)^{k}} .$$
Using Lemma \ref{lemma:reverse-markov},
$$ \prob \left[ \left| \left \{ i \in [L] | | f(y,z) - f(y,\Tilde{z}^{(i)})| \geq \dfrac{30 d^3 (2B)^k}{L}  \right \} \right| \geq \dfrac{L}{2 \times (2B)^k} \right] \geq \dfrac{1}{2L \times (2B)^k}  .$$

Observe that 
$$  \left| \left \{ i \in [L] | | f(y,z) - f(y,\Tilde{z}^{(i)})| \geq \dfrac{30 d^3 (2B)^k}{L}  \right \} \right| \geq \dfrac{L}{2 \times (2B)^k} \implies \displaystyle \sum_{i \in [L]} | f(y,z) - f(y,\Tilde{z}^{(i)}) | \geq 15d^3.$$
We conclude that for all $y \in \mathsf{GOOD}$, as $z, \tilde{z}^{(1)}, \cdots , \tilde{z}^{(L)}$ are sampled as in Procedure 1,
$$ \prob   \left[ \displaystyle \sum_{i \in [L]} | f(y,z) - f(y,\Tilde{z}^{(i)}) | \geq 15d^3 \right]  \geq \dfrac{1}{2L \times (2B)^k}. $$
Thus, with probability at least $\dfrac{1}{2L \times (2B)^k}$, procedure 1 returns a number greater than $15d^3 > 15d^2$.
\subsection*{Upper bound:}

Since $\E [(f(x) - g(x))^2] \leq \delta$ and $\prob_{y \in \cube^{[n] \setminus U}} [y \text{ is good}] \geq \mu/(3W)^k$, we have that

$$ \E [(f(x) - g(x))^2 | x_{[n] \setminus U} \text{ is good}] \leq \dfrac{\delta}{\mu} (3W)^k.$$ Now consider a uniformly sampled $y \in \mathsf{GOOD}$. Observe that as we sample $z \leftarrow \cube^U$ u.a.r, $\rho \leftarrow \{\rho_1, \cdots , \rho_{k+1}\}$ u.a.r and $\Tilde{z}^{(i)} \leftarrow N_{\rho, B_i} (z)$, the marginal distribution of $\tilde{z}^{(i)}$ is uniform on $\cube^U$. By Markov's inequality, we have
$$ \prob \left[(f(y,z) - g(y,z))^2 \geq \dfrac{1}{L^2} \right] \leq  \dfrac{L^2 \delta}{\mu} (3W)^k$$ 
and for all $i \in [L]$,

$$ \prob \left[(f(y,\tilde{z}^{(i)}) - g(y,\tilde{z}^{(i)}))^2 \geq \dfrac{1}{L^2} \right] \leq  \dfrac{L^2 \delta }{ \mu} (3W)^k.$$ 

By union bound, the probability that $ (f(y,z) - g(y,z))^2 \geq \dfrac{1}{L^2} $ or for some $i$, $(f(y,\tilde{z}^{(i)}) - g(y,\tilde{z}^{(i)}))^2 \geq \dfrac{1}{L^2}$ is at most $(L+1)  \dfrac{L^2  \delta }{ \mu} (3W)^k \leq \dfrac{2L^3 \delta}{ \mu} (3W)^k$. Our choice of $\delta$ ensures that this quantity is less than $< \dfrac{1}{2L \times (2B)^{k}}$. Observe that if none of these bad events holds, since the block sensitivity of $g$ is bounded above by $6d^2$ (Theorem \ref{thm:block_sensitivity}), we have that

$$ \displaystyle \sum_{i \in [L]} | g(y,z) - g(y, \tilde{z}^{(i)})| \leq 6d^2 \implies  \displaystyle \sum_{i \in [L]}  | f(y,z) - f(y, \Tilde{z}^{(i)})  | 
\leq 6d^2 + 1 < 15d^2.$$

Thus, we conclude
$$ \prob \left[ \displaystyle \sum_{i \in [L]} | f(y,z) - f(y, \Tilde{z}^{(i)})| > 15d^2 \right] < \dfrac{2L^3 \delta}{ \mu} (3W)^k <  \dfrac{1}{2L \times (2B)^k} .$$

As promised, we get conflicting lower and upper bounds for the probability that procedure 1 returns a number $>15d^2$. This is our desired contradiction.
\end{proof}

Now we show that we can improve the tail bound of \cite{DFKO06} under the additional assumption that $f$ has low degree. This follows straightforwardly from Theorem \ref{theorem:main_result}. 

\begin{theorem}
\label{theorem:improved_tail_bound_final}
There exists a universal constant $C>0$ such that the following is true: \newline
Let $f: \cube^n \rightarrow [0,1]$ be a degree $d$ function. Let $\theta = \dfrac{\var[f]^2}{d^C C^k}$ and $J = \{ j | \infl_j [f] \geq \theta \}$. If $\displaystyle \sum_{S \not \subseteq J} \hat{f}(S)^2 \geq \mu , $ then $\displaystyle \sum_{|S| > k} \hat{f}(S)^2 \geq \dfrac{\mu}{d^C C^k}$. 
\end{theorem}
\begin{proof}
Assume $\displaystyle \sum_{|S|>k} \hat{f}(S)^2 < \mu/2$ (otherwise we are done). Let $\tilde{C}$ be the universal constant from Theorem \ref{theorem:main_result}. \newline

The idea is to apply Theorem \ref{theorem:main_result} to the truncated function 
$$f^{\leq k} (x) = \displaystyle \sum_{|S| \leq k} \hat{f}(S) \chi_S(x).$$
Note that while $f^{\leq k}$ is not pointwise bounded, it satisfies $\E [(f^{\leq k})^2] \leq 1$ and $\infl_j [f^{\leq k}] \leq \infl_j [f]$ for all $j$ (this is clear from the Fourier expressions). Let $H = \{j | \infl_j[f^{\leq k}] \geq \theta\}$. We have $H \subseteq J$, so
$$ \displaystyle \sum_{S \not \subseteq H} \hat{f^{\leq k}}(S)^2 \geq \displaystyle \sum_{S \not \subseteq J} \hat{f}(S)^2 - \dfrac{\mu}{2} \geq \dfrac{\mu}{2}.$$

Applying Theorem \ref{theorem:main_result}, we get that for any bounded degree $d$ $g: \cube^n \rightarrow [0,1]$, $\E [(f(x) - g(x))^2] \geq \dfrac{\mu}{2d^{\tilde{C}} \tilde{C}^k}$. Taking $g$ to be our original function $f$, we get the desired tail lower bound:
$$ \E [(f - f^{\leq k})^2] \geq \dfrac{\mu}{2 d^{\tilde{C}} {\tilde{C}}^k} \implies \displaystyle \sum_{|S| > k} \hat{f}(S)^2 > \dfrac{\mu}{2 d^{\tilde{C}} {\tilde{C}}^k}.$$
Taking $C$ to be a slightly larger constant than $\tilde{C}$, we get that
$$ \displaystyle \sum_{|S|>k} \hat{f}(S)^2 \geq \dfrac{\mu}{d^C C^k}.$$
\end{proof}

\subsection{Random restrictions can be approximated by juntas}

In this section we use the fact that random restrictions have bounded tails to show that they can be approximated by juntas.

\begin{theorem}
    \label{theorem:approx_junta}
    For any constants $\tilde{C}_1, \tilde{C}_2 > 0$, there exist constants $\tilde{C}_3, \tilde{C}_4, \tilde{C}_5 > 0$ such that the following holds: \newline
    Let $f: \cube^n \rightarrow [0,1]$ be a degree $d$ polynomial and let $\rho$ be a random restriction with survival probability $\dfrac{\log(d)}{\tilde{C}_3 d}$. With probability at least $1-d^{-\tilde{C}_2}, $ $f_{\rho}$ is a $(d^{-\tilde{C}_1} \var [f], \var[f]^{-2} d^{\tilde{C}_4})$ junta. Moreover, if $J$ denotes the set of coordinates on which the junta depends, for each $j \in J$ we have $ \infl_j [f] \geq \var[f]^{-2} d^{-\tilde{C}_5}$.
\end{theorem}

\begin{proof}
    We consider a random restriction with survival probability $\dfrac{\log (d)}{\tilde{C}_3 d}$. \newline
    
    By Lemma \ref{lemma:bounded_tail}, the expected Fourier tail of $f_{\rho}$ above level $\log(d)$ is at most $\exp (-\tilde{C}_3 \log(d)/8) \var[f] = \dfrac{\var[f]}{d^{\tilde{C}_3/8}}.$ By Markov's inequality, with probability at least $1 - d^{\tilde{C}_3/16}$, the Fourier tail above $\log(d)$ is $\leq \dfrac{\var [f]}{d^{\tilde{C}_3/16}}$. Let $C$ be the constant from Theorem \ref{theorem:improved_tail_bound_final}.  Let $\mu = \dfrac{\var [f]}{ d^{\tilde{C}_3/16}} d^C C^{\log(d)} =  \dfrac{\var [f]}{ d^{\tilde{C}_3/16}} d^{2C}$, $\theta = \dfrac{\mu^2}{d^C C^{\log(d)}} = \dfrac{\mu^2}{d^{2C}}$ and $J = \{ j | \infl_j[f_{\rho}] \geq \theta\}$. Let $u: \cube^n \rightarrow [0,1]$ be the junta which reads the coordinates in $J$ and outputs the average over the coordinates in $J^c$. Choose $\tilde{C}_3$ large enough so that $\mu \leq d^{-\tilde{C}_1} \var[f]$. Applying Theorem \ref{theorem:improved_tail_bound_final}, we see that $u$ approximates $f_{\rho}$ to accuracy $ d^{-\tilde{C}_1} \var[f]$. Using the fact that total influence is bounded by $d$, we see that $u$ has arity $\leq \var[f]^{-2} d^{ C' \tilde{C}_3}$ for a universal constant $C'$. Taking $(\tilde{C}_4,\tilde{C}_5)= (C' \tilde{C}_3, \tilde{C}_3/32 - 2C)$, we are done.
    
\end{proof}

\subsection{Aaronson-Ambainis conjecture is true for random restrictions}

\begin{theorem}
\label{theorem:aa_for_restrictions}
    There exist constants $C_1,C_2>0$ such that the following holds: let $f: \cube^n \rightarrow [0,1]$ be a degree $d$ polynomial $(d \geq 2)$ with $\var[f] \geq 1/d$. Let $\rho$ denote a random restriction with alive probability $\dfrac{\log(d)}{ C_1 d}$. Then,
    $$ \prob \left[f_{\rho} \text{ has a coordinate with influence }\geq \dfrac{\var[f]^2}{d^{C_2}}\right] \geq \dfrac{\var[f] \log(d)}{50C_1d}.$$
    
\end{theorem}
\begin{proof}
    Let $M$ be a large constant. Apply Theorem \ref{theorem:approx_junta} with $(\tilde{C}_1,\tilde{C}_2)= (M, M)$ to get constants $\tilde{C}_3, \tilde{C}_4, \tilde{C}_5$. Let $\rho$ be a random restriction with survival probability $\dfrac{\log (d)}{\tilde{C}_3 d}$.
    By Lemma \ref{lemma:variance_of_random_restriction},
    $$ \E [\var[f_{\rho}]] \geq \dfrac{\var[f] \log(d)}{\tilde{C}_3 d}$$
    so by Lemma \ref{lemma:reverse-markov},
    $$ \prob \left[\var[f_{\rho}] \geq \dfrac{\var[f] \log(d)}{2\tilde{C}_3d}\right] \geq \dfrac{\var[f] \log(d)}{2 \tilde{C}_3 d}.$$
     Since $\var[f] \geq 1/d$, $d^{-M} \leq \dfrac{\var[f] \log(d)}{10 \tilde{C}_3 d}$. By Theorem \ref{theorem:approx_junta} and Remark \ref{remark:junta_approximation}, with probability at least $1-d^{-M}$, there exists a $J_{\rho} \subseteq [n]$ such that every coordinate in $J_{\rho}$ has influence $\geq \var[f_{\rho}]^{-2} d^{- \tilde{C_5}} $ and 
    \begin{align*} \displaystyle \sum_{S \not \subseteq J_{\rho}} \hat{f_{\rho}}(S)^2 & \leq d^{-M} \var [f] .\end{align*}
    So with probability at least $\dfrac{\var[f] \log(d)}{2\tilde{C}_3 d} - d^{-M} \geq \dfrac{\var[f] \log(d)}{4 \tilde{C_3}d}$, both these events (high variance of $f_{\rho}$ and existence of $J_{\rho}$) hold and we have that
    \begin{align*}
     \displaystyle \sum_{S \subseteq J_{\rho}} \hat{f_{\rho}}(S)^2 & \geq \var[f_{\rho}] - d^{-M} \var[f] \geq \dfrac{\var[f] \log(d)}{4 \tilde{C}_3 d}.  \end{align*}
    In particular, we have that $J_{\rho} \neq \phi$. Since for each $j \in J_{\rho}$ we have $\infl_j [f_{\rho}] \geq \var[f_{\rho}]^{-2} d^{-\tilde{C}_5}$, we are done by taking $(C_1, C_2)= (\tilde{C}_3, 2+\tilde{C}_5)$.
\end{proof}
\section{Conclusions and further directions}
In this paper, we showed that if $f: \cube^n \rightarrow \{0,1\}$ is a degree $d$ polynomial, a large fraction of its random restrictions have an influential coordinate. We observe that this implies one of the results proven in \cite{LZ19} about the existence of small sensitive blocks with a slightly different set of parameters. \newline

Let $f: \cube^n \rightarrow [0,1]$. An input $x \in \cube^n$ is said to be $(r,\epsilon)$ sensitive if there exists a $y$ such that $d(x,y) \leq r$ and $|f(x)-f(y)| \geq \epsilon$. \cite{LZ19} proves the following:

\begin{theorem}
    \label{lovett-zhang}
    If $f: \cube^n \rightarrow [0,1]$ has degree $d$, then at least $\Omega (\var[f])$ fraction of the inputs are $(r,\epsilon)$ sensitive where $\epsilon= \text{poly}(\var[f]/d), r= \text{poly}(d, 1/\epsilon, \log(n))$
\end{theorem}

An immediate consequence of our result is that at least $\Omega(\var[f]/d^{O(1)})$ fraction of inputs are $(1, \epsilon)$ sensitive where $\epsilon = \text{poly}(\var[f]/d)$. Thus, while we lose a bit in the fraction of sensitive inputs, we gain by letting our block size be exactly 1 instead of $\text{poly}(d, 1/\epsilon, \log (n))$.

It would be interesting to see if we can extend this to the full Aaronson-Ambainis conjecture. We describe a potential approach here.

\begin{itemize}
    \item Given a degree $d$ polynomial $f: \cube^n \rightarrow [0,1]$, we can lift it with a Boolean function $g: \cube^m \rightarrow \cube^n$ each of whose coordinates $g_i$ is unbiased and given by a low degree function. Then, the lifted polynomial $f \odot g: \cube^m \rightarrow [0,1]$ will be a low degree polynomial. As long as the $g_i$'s are pairwise independent, the variance of $f$ will be preserved as well. Our result shows that a large fraction of random restrictions of $f \odot g$ have an influential coordinate. Can we construct $g_1, g_2, \cdots , g_n$ appropriately such that this allows us to conclude $f$ must have an influential coordinate as well? The $g_i$'s should introduce correlations between the different input bits of $f$ so that most random restrictions of $f \odot g^m$ \textit{look the same} in some appropriate sense.  
\end{itemize}

\bibliographystyle{alpha}
\bibliography{main}

\end{document}